\documentclass[11pt]{article}
\usepackage{mathtools}
\usepackage{amsmath}
\usepackage{amsthm}
\usepackage{amssymb}
\usepackage{algorithm}
\usepackage{subfig}
\usepackage{color}
\usepackage[english]{babel}
\usepackage{graphicx}
\usepackage{wrapfig,epsfig}
\usepackage{epstopdf}
\usepackage{url}
\usepackage{graphicx}
\usepackage{color}
\usepackage{epstopdf}
\usepackage{algpseudocode}
\usepackage{scrextend}
\usepackage[T1]{fontenc}
\usepackage{bbm}
\usepackage{comment}
\usepackage{authblk}
\usepackage{multicol}

\usepackage{tikz}
\usepackage{hyperref}  
\hypersetup{colorlinks=true,citecolor=blue,linkcolor=blue} 
\usetikzlibrary{arrows}
\usepackage[margin=1in]{geometry}
\graphicspath{{./figs/}}

\author{Vasileios Nakos\thanks{Saarland University and Max Planck Institute for Informatics, Saarland Informatics Campus, Saarbr\"ucken, Germany. \texttt{vnakos@mpi-inf.mpg.de}.  This work is part of the project TIPEA that has received funding from the European Research Council (ERC) under the European Unions Horizon 2020 research and innovation programme (grant
agreement No. 850979). Part of this work was done while the author was a ph.D. student at Harvard University and supported by NSF CAREER award CCF-1350670.}}

\affil{\texttt{vnakos@mpi-inf.mpg.de,}\\ Saarland University and Max Planck Institute for Informatics.}

\date{}

\title{Nearly Optimal Sparse Polynomial Multiplication}
\newtheorem{theorem}{Theorem}[section]
\newtheorem{lemma}[theorem]{Lemma}
\newtheorem{definition}[theorem]{Definition}

\newcommand{\wt}{\widetilde}

\renewcommand{\varepsilon}{\epsilon}
\renewcommand{\tilde}{\wt}

\makeatletter
\newcommand*{\RN}[1]{\expandafter\@slowromancap\romannumeral #1@}
\makeatother

\usepackage{lineno}

\usepackage{etoolbox}

\makeatletter

\newcommand{\define}[4][ignore]{%
  \ifstrequal{#1}{ignore}{}{
  \@namedef{thmtitle@#2}{#1}}%
  \@namedef{thm@#2}{#4}%
  \@namedef{thmtypen@#2}{lemma}%
  \newtheorem{thmtype@#2}[theorem]{#3}%
  \newtheorem*{thmtypealt@#2}{#3~\ref{#2}}%
}

\newcommand{\state}[1]{%
  \@namedef{curthm}{#1}
  \@ifundefined{thmtitle@#1}{
  \begin{thmtype@#1}
    }{
  \begin{thmtype@#1}[\@nameuse{thmtitle@#1}]
  }
    \label{#1}
    \@nameuse{thm@#1}
  \end{thmtype@#1}
  \@ifundefined{thmdone@#1}{
  \@namedef{thmdone@#1}{stated}%
  }{}
}

\newcommand{\restate}[1]{%
  \@namedef{curthm}{#1}
  \@ifundefined{thmtitle@#1}{
    \begin{thmtypealt@#1}
    }{
  \begin{thmtypealt@#1}[\@nameuse{thmtitle@#1}]
  }
    \@nameuse{thm@#1}
  \end{thmtypealt@#1}
  \@ifundefined{thmdone@#1}{
  \@namedef{thmdone@#1}{stated}%
  }{}
}

\newcommand{\thmlabel}[1]{
  \@ifundefined{thmdone@\@nameuse{curthm}}{\label{#1}
    }{\tag*{\eqref{#1}}}
}
\makeatother

\begin{document}

\begin{titlepage}
  \maketitle
  \begin{abstract}

In the sparse polynomial multiplication problem, one is asked to multiply two sparse polynomials $f$ and $g$ in time that is proportional to the size of the input plus the size of the output. The polynomials are given via lists, $F$ and $G$, of their coefficients. Cole and Hariharan (STOC 02) have given a nearly optimal algorithm when the coefficients are positive, and Arnold and Roche (ISSAC 15) devised an algorithm running in time proportional to the ``structural sparsity'' of the product, i.e. the set $\mathrm{supp}(F) + \mathrm{supp}(G)$. The latter algorithm is particularly efficient when there are not ''too many cancellations'' of coefficients in the product. In this work we give a clean, nearly optimal algorithm for the sparse polynomial multiplication problem.

  \end{abstract}
  \thispagestyle{empty}
\end{titlepage}

\section{Introduction}

Multiplying two polynomials is a fundamental computational primitive, with multiple applications in computer science. Using the Fast Fourier Transform, one can perform multiplication of polynomials stored as vectors of floating point numbers in time $O(n \log n)$, where $n$ is a bound on the largest degree.

An important and natural question is whether, and under which circumstances, a faster algorithm can be invented. Researchers have tried to obtain algorithms that beat the $O(n \log n)$-time bound, when the two polynomials are sparse, i.e. the number of non-zero terms in each polynomial is at most $k$. Interestingly, some ideas from the relevant literature have found applications in computer algebra packages such as Maple, Mathematica and Singular, including ways to represent and store polynomials \cite{maza2001sparse,monagan2014poly,monagan2015design,groves2016sparse}. 

When two polynomials have at most $s$ coefficients, the trivial algorithm gives $O(s^2\log n \log s)$ time, which is already and improvement for $s \leq \sqrt{n/ \log n}$. With the rise of big data and the exploitation of inherent sparsity in the data sets by researchers from compressed sensing \cite{crt06, crt06a, ct06, hikp12a} and machine learning, the question of whether two sparse polynomials can be rapidly multiplied becomes a fundamental and natural question. The desirable  goal is to obtain an algorithm that is output-sensitive, i.e. runs in nearly linear time with respect to $k+s$, where $k$ is the number of non-zero coefficients in the product. A result of Cole and Hariharan \cite{cole2002verifying} obtains an algorithm that runs in $O( k \log ^2 n)$ time, when the coefficients of the two polynomials are non-negative. A data structure for carefully allocating and de-allocating memory has been designed in \cite{yan1998geobucket}, trying to tackle the problem of memory handling can be the main bottleneck in complexity of sparse multiplication in practical scenarios. The aforementioned algorithm is based on a heap, an idea which also lead to implementations developed in \cite{monagan2007polynomial,monagan2009parallel,monagan2014poly,monagan2015design}. The authors in \cite{monagan2009parallel} develop a parallel algorithm for multiplying sparse distributed polynomials, where each core uses a heap of pointers to multiply parts of polynomials, exploiting its L3 cache. The same authors in \cite{monagan2014poly} have created a data structure suitable for the Maple kernel, that allows for obtains significant performance in many Maple library routines.

When the support of the product is known or structured, work in \cite{roche2008adaptive,roche2011chunky,van2012complexity,van2013bit} indicates how to perform the multiplication fast. Using techniques from spare interpolation, Arnold and Roche \cite{arnold2015output} have given an algorithm that runs in time that is nearly linear in the ``structural sparsity'' of the product, i.e. the sumset of the supports of the two polynomials. When there are not ``too many'' cancellations, this is roughly the same as the size of the support of the product, and the above algorithm is quite efficient. However, in the presence of a considerable amount of cancellations in the product, the aforementioned algorithm becomes sub-optimal. Removing this obstacle seems to be the final step, and has been posed as an open problem in the excellent survey of \cite{roche2018can}. 

In this paper, we resolve the aforementioned open question, giving an algorithm that is near-linear in the size of the input plus the size of the output. This can be considered near-optimal. Moreover, assuming optimality of the Fast Fourier Transform, one cannot expect an algorithm faster than $O((k+s)\log (s+k))$, and our algorithm reaches quite close to this barrier. We note that one can use the big hammers of filter-based (sparse) Fourier sampling \cite{gms05,hikp12a,hikp12b,k16,k17} along with semi-equispaced Fourier transforms \cite{dutt1993fast} to obtain nearly optimal algorithms for sparse polynomial multiplication, but these algorithms are much more complicated and technical, and thus likely to be worse in a real-life scenario. Furthermore, puting all the pieces together is a non-trivial task for someone who is not an expert in Sparse Fourier transform, as how to solve sparse polynomial multiplication using those techniques has not been explicitly written down in the literature. In contrast, our algorithm is much less technically demanding and we feel accessible even to a graduate level student.

\section{Preliminaries}

We will be concerned with polynomials with integer coefficients. This suffices for most applications, since numbers in a machine are represented using floating point arithmetic. We denote by $\mathbb{Z}_n$ the ring of residue modulo $n$ and by $[n]$ the set $\{0, 1, \ldots n-1\}$. For a vector $x \in \mathbb{R}^n$ we define $\mathrm{supp}(x) = \{ i \in [n]: x_i \neq 0\}$ and $\|x\|_0 = |\mathrm{supp}(x)|$. All vectors in this paper are zero-indexed. 


We define the cyclical convolution of two vectors $x,y \in \mathbb{R}^n$ as the $n$-dimensional vector $x\star y$ such that 
\[(x \star y)_i = \sum_{j,j'\in [n]:  (j+j')~\mathrm{mod}~n =i} x_j y_{j'}.\] 

It is easy to see that cyclical convolution is immediately related with polynomial multiplication: if $f(z) = \sum_{j=0}^n \alpha_jz^j$ and $g(z) = \sum_{j=0}^n \beta_j z^j$, we have that the polynomial $(f \cdot g) (z) =  \sum_{j=0}^{2n} \gamma_j z^j$ satisfies $\gamma = \alpha \star \beta$, where $\gamma = (\gamma_0,\gamma_1,...\ldots,\gamma_N) \in \mathbb{Z}^{N}, \alpha = (\alpha_0,\alpha_1,\ldots,\alpha_n, 0, \ldots, 0) \in \mathbb{Z}^{N}, \beta = (\beta_0,\beta_1,\ldots,\beta_n, 0, \ldots, 0) \in \mathbb{Z}^{N}$, for $N = 2n$. 

Throughout the paper we assume that we work on a machine where the word size is $w = C \log n = \Theta( \log n)$ for some sufficiently large constant $C$, all coefficients of the polynomials fit in a word, and elementary operations between two integers given as part of the input can be done in $O(1)$ time. For a complex number $z$ we denote by $|z|$ its magnitude, and by $\mathrm{arg}(\phi)$ its phase. We also use $\tilde{O}(f)$ throughtout the paper to denote $ f \cdot \mathrm{poly}(\log f)$.

\section{Result}

The contribution of this work is the following.

\begin{theorem}\label{thm:sparse_multiplication_theorem}
Let $u,v \in \mathbb{Z}^n$, given as lists of their non-zero coordinates along with their values. Let $N = 2n$. Define $x = (u_0,u_1,\ldots,u_{n-1},0,\ldots,0 ) \in \mathbb{Z}^N$, as well as 
$y = (v_0,v_1,\ldots,v_{n-1},0,\ldots,0 ) \in \mathbb{Z}^N$. Let $s = \|x\|_0 + \|y\|_0$ (size of input), and $k = \|x \star y \|_0 +4$ (size of output). Then, with probability $99/100$, we can compute a list which contains the non-zero coefficients and values of $x \star y$, in time $O( (s+k) \cdot \mathrm{poly}(\log N))$. In particular, the running time is (ignoring constant factors)

\begin{align*}
 k \log^2 N \cdot \log( k \log^2 N) \cdot \log \log k + \\
s \cdot \log k \log N \cdot (\log \log k)^2  + \\
\tilde{O}(\log^3 N) \cdot \log \|x \star y\|_0.
\end{align*}

\end{theorem}

\section{Overview of the Algorithm}

Our algorithm resembles the iterative framework employed in the renowned sublinear-time sparse recovery algorithm of \cite{glps12}. Similarly to that paper, our algorithm implements a routine which carefully hashes every coordinate $i \in [N]$ to an appropriate number of buckets, and from each bucket identifies the location and values of at least a constant fraction of the non-zero coordinates in $x \star y$. Let $r$ be the $k$-sparse vector obtained. Considering vector $(x\star y) - r$, we can perform a similar ``hash and identify'' procedure to peel of a constant fraction of the coordinates left. Repeating in $O( \log k)$ phases, we can guarantee that we've found all coordinates and cleaned up all mistakes introduced.

The main contribution of our work lies in effectively hashing the support of $x \star y$ and identifying the locations and the values of the non-zero coordinates, given access only to $x$ and $y$. Let us assume for now that $k$ is given to us as a promise. It is a folklore fact (see for example Section 8 in \cite{CL15}) that if one takes a random prime $p \in [10 k \log^2 N]$, folds $x$ by forming a $p$-length vector by summing all $x_j$ for $j$ which are equal modulo $p$, folds $y$ in the same way, and then computes the cyclical convolution of the folded vectors, then most coordinates in the support of $x \star y$ are isolated in the resulting vector. However, this trick only gives the values of the non-zero coordinates in $x \star y$, and it is unclear how to extract the corresponding indices. What we observe is that for any (complex) number $g$, the following thing holds. If we first multiply each $x_j$ by $g^j$ to obtain vector $x'$, then fold $x'$ as before, do the same thing with $y$, and then compute the cyclical convolution of the folded vectors, we can prove that the vector we obtain has in its $i$th entry the number

	\[	\sum_{j \in [N]: j~\mathrm{mod}~p = i}	(x \star y)_j \cdot g^j.		\]

The nice thing now is that if some $j \in \mathrm{supp}(x \star y)$ is isolated under the hash function $h(x) = x~\mathrm{mod}~p$, we will have $(x \star y)_j \cdot g^j$ in hand, rather than only $(x \star y)_j$ that we had before. From this we would like to infer both $(x \star y)_j$ and  $j$. It is easy to see that even for $(x \star y)_j = 1$ this is essentially a discrete logarithm problem. By choosing $g \neq 1$ to be a real number, this would result in $N$-digit numbers, which take too long to manipulate. The right solution is to choose $g = \omega$, where $\omega$ is an $(2\cdot N)$-th root of unity rounded to fit in a word. In this way, every $g^j$ can still be written down using $\Theta(\log n)$ digits, and we can infer integer $j \in [N]$ from $g^j$ by performing a ternary search over the corresponding quarter of the complex circle. Note that we chose $(2N)$-th root of unity in order to infer both $(x \star y)_j$ and $j$ from $(x \ast y)_j \cdot g^j$: if the phase of the aforementioned complex number is less than $ \frac{\pi}{2}$, then $(x \star y)_j$ is positive (and vice versa), otherwise it is negative. In the latter case, to read the correct location one needs to decrease the found $j$ by $N$.

Of course, the above discussion assumes that we know $k$. We can guess $k$ by doubling and invoking standard fingerprinting techniques. Putting carefully everything together results in the desired guarantee claimed in \ref{thm:sparse_multiplication_theorem}.

\section{Algorithm and Proof}

We proceed by building the tools needed for the proof of theorem \ref{thm:sparse_multiplication_theorem}. In what folows $\omega$ is an $(2N)$-th root of unity; our algorithm should treat it as rounded in order to fit in a word. We remind that our word size is $\Omega( \log n)$, which is the word size needed just to write down the largest exponent of the polynomial.

The following operator will be particularly important for our algorithm.
 \begin{definition}\label{def:operator} Let $x \in \mathbb{Z}^N$. For each $m$ define the function $h_m: [N] \rightarrow [m]$ by 
\[ h_{m}(i) = i ~\mathrm{mod}~ m.\] 

Moreover, define $\mathcal{P}_{m}(x) \in \mathbb{Z}^m$ to be such that

	\[	\left( \mathcal{P}_{m}(x) \right)_i = \sum_{j \in [N]: h_{m}(j) = i} x_j \cdot \omega^{ j}, \forall i \in [m].	\] 

We remind that $\omega$ is a $(2N)$-th root of unity rounded to fit in a word.

\end{definition}

The following claim lies at the crux of our argument.
\begin{lemma} \label{lem:prepare}
Given vectors $x,y,w \in \mathbb{Z}^N$ the vector $\mathcal{P}_{m}((x\star y) - w )$ up to $1/\mathrm{poly}(N)$ error can be computed in time \[O( \left( \|x\|_0 + \|y\|_0 + \|w\|_0\right) \log N + m \log m).\]
\end{lemma} 
\begin{proof}

First, note that 
\[\mathcal{P}_{m}((x\star y) - w )=  \mathcal{P}_{m}(x \star y) - \mathcal{P}_{m}(w),	\]
since $\mathcal{P}_m$ is a linear operator.
We compute $ \mathcal{P}_{m}(x), \mathcal{P}_{m}(y), \mathcal{P}_{m}(w)$ in time $O(m+ \left(\|x\|_0 + \|y\|_0 +\|w\|_0\right) \log N)$ by computing $\omega^j$ for all $j \in \mathrm{supp}(x) \cup \mathrm{supp}(y) \cup \mathrm{supp}(w)$ using Taylor expansion of sine and cosine functions and keeping the first $\Theta(\log N)$ digits. We then compute, via FFT in time $O(m \log m)$, the vector $ \mathcal{P}_{m}(x) \star \mathcal{P}_{m}(y)$. We claim that 
\[	|\left( \mathcal{P}_{m}(x\star y) \right)_i - \left( \mathcal{P}_{m}(x) \star \mathcal{P}_{m}(y) \right)_i| \leq 1/\mathrm{poly}(N).	\]
 
 If we could write down the roots of unity using using an infinite number of bits we would have

\[	\left( \mathcal{P}_{m}(x\star y) \right)_i = \left( \mathcal{P}_{m}(x) \star \mathcal{P}_{m}(y) \right)_i.	\] 

 The only catch is that we have to round every root of unity to $\Theta(\log N)$ bits. We will prove the latter inequality and notice that the precision issue arises only in the usage of the identiy $\omega^j \cdot \omega^{j'} = \omega^{j + j'}$; when rounding to $\Theta(\log N)$ bits this introduces a negligible error of at most $1/\mathrm{poly}(N)$ which allows our argument to go through. Thus, we focus on proving the latter claim, via the following series of inequalities: 

\begin{flalign}
&\left( \mathcal{P}_{m}(x) \star \mathcal{P}_{m}(y)  \right)_i   = \\
&\sum_{\ell,\ell' \in [m]: (\ell+\ell')~\mathrm{mod}~m = i } (\mathcal{P}_{m}(x))_\ell \cdot (\mathcal{P}_{m}(y))_{\ell'}   = \\
&\sum_{\ell,\ell' \in [m] : (\ell+\ell') ~\mathrm{mod}~m = i} \left( \sum_{j,j' \in [n]: h_{m}(j) = \ell, h_{m}(j') = \ell'} x_j y_{j'} \omega^{j+j'}\right)= \\
&\sum_{\ell,\ell' \in [m], j,j' \in [N]: h_{m}(j) = \ell, h_{m}(j') = \ell', (\ell+\ell')~\mathrm{mod}~m = i} x_j y_{j'} \omega^{j+j'} = \\
&\sum_{j,j' \in [N]: (h_{m}(j) +  h_{m}(j')) ~\mathrm{mod}~m = i}  x_j y_{j'} \omega^{j+j'}= \\
&\sum_{j,j' \in [N]: h_{m}( (j+j')~\mathrm{mod}~m))  = i}  x_j y_{j'} \omega^{j+j'}= \\
&\sum_{j'' \in [N], h_m (j''~\mathrm{mod}~m)= i }   \sum_{j,j' \in [N]: j+j' = j''} x_j y_{j'} \omega^{ j'' }=\\
&\sum_{j'' \in [N], h_m (j''~\mathrm{mod}~m)= i }  \omega^{j''} \left(  \sum_{j,j' \in [N]: j+j' = j''} x_j y_{j'} \right)=\\
&\sum_{j'' \in [N], h_m(j'') = i } \omega^{j''} \left( x \star y \right)_{j''},
\end{flalign}

where (1) to (2) follows by defition of convolution, (2) to (3) by definition of the $\mathcal{P}_B$ operator, (3) to (4) by expanding the product in (2), (5) to (6) by the trivial fact each element in $[n]$ is mapped to some element in $[m]$ via $h_m$, (6) to (7) by the fact that $(h_m(a) + h_m(b) )~\mathrm{mod}~m = ( a~\mathrm{mod}~m + b~\mathrm{mod}~m)~\mathrm{mod}~m = (a+b)~\mathrm{mod}~m = h (a+b) $, (7) to (8) by introducing the auxilliary variabe $j''=j+j'$, (8) to (9) by the fact that $\omega^{j''}$ can be pulled outside of the inner sum since in that scope $j''$ is fixed, and (9) to (10) since $h_m(j'' ~\mathrm{mod}~m ) = (j''~\mathrm{mod}~m) ~\mathrm{mod}~m = (j''~\mathrm{mod}~m) = h_m(j'')$ and the fact that the inner sum is the definition of convolution evaluated at point $j''$.
\end{proof}

In what follows $C$ is some sufficiently large absolute constant.

\begin{algorithm} 
\caption{$\textsc{Locate}(x,y,w,B,\delta)$}  
\begin{algorithmic}
\State $L \leftarrow \emptyset$
\For { $t \in [ 5\lceil \log(1/\delta) \rceil ]$ }
	\State Pick random prime $p$ in $[CB\log^2 N]$.
	\State Compute $\mathcal{P}_{p}((x\star y) - w)$, using Lemma \ref{lem:prepare}.
	\If { $ |(\mathcal{P}_{p}((x\star y) - w))_b| \geq 1/\mathrm{poly}(n)$ for more than $B$ values $b \in [p]$}
	\State Return $\vec{0} \in \mathbb{R}^N$
	\EndIf
	\For{ $b \in [p]$ }
		\If {$ |(\mathcal{P}_{p}((x\star y) - w))_b| \geq 1/\mathrm{poly}(n)$}
			\State $v \leftarrow |(\mathcal{P}_{p}((x\star y) - w))_b|$ \label{lin:tricky}
			\State $ar \leftarrow (\mathcal{P}_{p}((x\star y) - w))_b / |(\mathcal{P}_{p}((x\star y) - w))_b|$.
			\State Compute $i$ from $ar$ using Lemma \ref{lem:binary_search}
			\If { $i> N$ }
				\State $v \leftarrow -v$ 
				\State $i \leftarrow i - N$
			\EndIf
			\State $L \leftarrow L \cup \{ (i,v) \}$
		\EndIf
	\EndFor
\EndFor

\State Prune $L$ to keep pairs  $(i,v)$, which appear at least $(3/4)\cdot 5 \log(1/\delta)$ times.
\State $z \leftarrow \vec{0} \in \mathbb{R}^N$
\For{ $(i,v) \in L$}
	\State $z_i \leftarrow v$
\EndFor
\State Return $z$
\end{algorithmic}
\end{algorithm}

\begin{algorithm}  \label{alg:hash_and_iterate}
\caption{$\textsc{HashAndIterate}(x,y,B,\delta)$} 
\label{alg:hash_and_iterate}
\begin{algorithmic}
\State $w^{(0)} \leftarrow 0$
\For{ $r= 1$ to $ \lceil \log B \rceil$}
	\State $\delta_r \leftarrow \delta / \log B$
	\State $B_r \leftarrow B \cdot 2^{-r+1} $
	\State $z\leftarrow \textsc{Locate}(x,y,w^{(r)},B_r,\delta_r)$
	\State $w^{(r)} \leftarrow w^{(r)} + z$
\EndFor
\State Return $w^{(\lceil \log B \rceil)}$
\end{algorithmic}
\end{algorithm}

The following Lemma is important, since we are dealing with numbers with finite precision.
\begin{lemma} \label{lem:approximate_root_unity}
Let $a,b \in [2N]$ with $a \neq b$. If $\omega$ is rounded such that it fits in a word, then $|\omega^a - \omega^b| = \Omega(1/ N )$.
\end{lemma}
\begin{proof}
The quantity is minimized when $a=b+1$. For $N$ sufficiently large, it can then be approximated by an arc of length $2\pi /(2N) $, and since the word size $w$ is $\Omega(\log N)$ we get the desired result.
\end{proof}

The follows Lemma is a crucial building block of our algorithm.
\begin{lemma}\label{lem:binary_search}
Given $\omega^j$ for $j \in [2N]$ with $\omega$ rounded to fit a word, one can find $j$ in time $O(\log^2 N)$.  
\end{lemma}

\begin{proof}
From the pair (real part of $\omega^j$,imaginary part of $\omega^j$) we can find in which of the four following sets $j$ lies in 
\begin{align*}
 \left\{0,\ldots , \lceil N \right\},\\
 \left\{\left\lceil \frac{N}{2} \right\rceil+1,\ldots, N \right\}\\
 \left\{N +1,\ldots,\left\lceil \frac{3}{2}N \right\rceil\right\}\\
\left\{\left\lceil \frac{3}{2}N \right\rceil +1,\ldots,2N-1\right\}
\end{align*} since each one corresponds to an arc of length approximately $\pi/4$ of the complex circle. After detecting the set (equivalently the corresponding arc of the complex circle) one can perform a standard ternary search to find $j$. Due to Lemma \ref{lem:approximate_root_unity} $O(\log N)$ iterations suffice to find $j$. In order for to implement ternary search, at every time we need to compute $\omega^\ell$, for some $\ell \in [2N]$,  which can be done in $\Theta(\log N)$ time up to error $1/\mathrm{poly}(N)$ by performing a Taylor expansion on $\omega^\ell$. This gives the desired result. We note that if we are allowed to precompute all $(2N)$-th roots of unity we may obtain running time $O(\log N)$.
\end{proof}

Before proceeding, we first note the following catch in Line \ref{lin:tricky}. In fact, $v$ will be a complex number with integer magnitude plus (due to rounding errors) a small $1/\mathrm{poly}(n)$ error. Thus we can round $v$ to the closest integer to read off exactly the desired value. 

The following Lemma resembles standard isolation-type arguments which appear in the sublinear-time sparse recovery literature. $C$ is a large enough absolute constant.

 \begin{lemma} \label{lem:isolation}
Let an integer $B$ such that $B > C\cdot \| (x \star y) - w \|_0$, and let $p$ be chosen at random from $[CB \log^2 n]$. Then, with probability $1-q$, there exist at least $(1-\gamma) \| (x \star y) - w)\|_0$ indices $j \in \mathrm{supp}((x \star y) -w )$ such that 
	\[	\forall j' \in \mathrm{supp}( (x \star y) -w ) \setminus \{j\} : h_{p}(j') \neq h_{p}(j),\]

where $2 C^{-2}/q  = \gamma $.
\end{lemma}

\begin{proof}
Let $j,j' \in \mathrm{supp}((x \star y) - w)$, with $j \neq j'$. The hash function $h_{p}$ is not pairwise independent, but the following property, which suffices for our purpose, holds
	\[	\mathbb{P} \left[ h_{p} (j) = h_{p}(j') \right] \leq 1/B.\]

To see that, observe first that in order for $h_{p} (j) = h_{p}(j')$ to hold, it must be the case that $ p$ is a divisor of $j-j'$. Since $j-j' \leq N$ there there can be at most $\lceil \log N\rceil$ prime divisors of $j-j'$, otherwise $j -j'$ would be at least $2^{\left\lceil  \log N \right\rceil+1 } > N$. By the Prime Number Theorem, there exist at least $(C/2) B \log N$ primes in $[C B \log^2 N]$, and hence a random prime will be one of the divisors of $j-j'$ with probability $2/(CB)$. All the primes in $[C B \log^2 N]$ can be found in $O( B \log^2 N \cdot\log \log (B \log^2 N) )$ time using Eratosthene's sieve, and then we can sample uniformly at random from them.


Le the random variable $X_j$ be the indicator variable of the event 
\[ \mathcal{E}_j  \vcentcolon \{\exists j' \in \mathrm{supp}((x \star y) -w) \setminus \{j\} : h_{p}(j) = h_{p}(j')\}.  \]

	Its expected value is $\mathbb{E} \left[ X_j \right ] = \mathbb{P} \left[ \mathcal{E}_j~\mathrm{holds} \right]\leq (\|(x\star y ) - w)\|_0 -1 ) \cdot (2/CB) \leq 2C^{-2}$, by a union-bound. We have that  
	\[	\mathbb{E} \left[ \sum_{j \in \mathrm{supp}((x \star y) - w)} X_j \right] \leq (2C^{-2}) \|(x \star y) - w \|_0.\]

By Markov's inequality, with probability $1-q$ there exist at most $\gamma \|(x \star y) - w\|_0$ indices $j \in \mathrm{supp}((x\star y) - w )$ such that $X_j = 1$, if $ 2C^{-2}/q =  \gamma $. This finishes the proof of the claim.
\end{proof}


We proceed by analyzing the iterative loop in Algorithm \ref{alg:hash_and_iterate}.

\begin{lemma}
Let the constants $C,\gamma,q$ be as in Lemma \ref{lem:isolation} with $q \leq 2^{-12/5}$, and assume that $B > C \|(x\star y) -w \|_0$. If $ (x \star y) - w$ is not the zero vector, then with probability $1-\delta$ the subroutine $\textsc{Locate}(x,y,w,B,\delta)$ returns a vector $z$ such that
	\[	\|z - ((x \star y) - w)\|_0 \leq (5\gamma) \|(x \star y) - w )\|_0.	\]
\end{lemma}

\begin{proof}
Fix $t \in [5 \log(1/\delta)]$, and assume that $q,C,\gamma$ satisfy $2C^{-2}/ p = \gamma$ 
We have that 

\begin{align*}
(\mathcal{P}_{p}(x \star y) )_i - (\mathcal{P}_{p}(w) )_i = (\mathcal{P}_{p}(x \star y -w ) )_i = \\
\sum_{j \in [N]: h_{p}(j)=i} ((x \star y) - w)_j \omega^j = \\
\sum_{j \in [N]:  h_{p}(j)=i~\mathrm{and}~((x \star y) -w )_j \neq 0}   ((x \star y) - w)_j\omega^j 
\end{align*}

The condition of Lemma \ref{lem:isolation} holds, so with probability $1-q$ its conclusion also holds. Condition on that event and consider the at least $(1-\gamma)\|(x \star y) -w \|_0$ indices in $\|(x\star y) - w \|_0$, for which the conclusion of Lemma \ref{lem:isolation} holds. Fix such an index $j^*$ and let $i^* = h_{p}(j^*)$. Due to the isolation property, we have that
\begin{align*}
(\mathcal{P}_{p}(x \star y) )_{i^*} - (\mathcal{P}_{p}(w) )_{i^*}  =  ((x \star y)_{j^*}  - w_{j^*} )\omega^{j^*}.
\end{align*}

Now, due to Lemma \ref{lem:approximate_root_unity} subroutine $\textsc{Locate}(x,y,w,B,\delta)$ will infer $j^*$ correctly from $(\mathcal{P}_{\sigma,B}(x \star y))_{i^*} - (\mathcal{P}_{\sigma,B}w)_{i^*}$, as well as $(x \star y)_{j^*} - w_{j^*}$. We will say $j^*$ is recognised in repetition $t$.

For the rest of the proof, unfix $t$. Since the conclusion of Lemma \ref{lem:isolation} holds with probability $1-q$, the number of $t \in [5 \log (1/\delta)]$ for which the conclusion of the Lemma holds is at least $4 \log(1/\delta)$ with probability $1-\delta$ since

	\[	{ 5 \log(1/\delta) \choose (5/2) \log(1/\delta)} p^{(5/2) \log(1/\delta)} \leq 2^{ 5 \log(1/\delta)} q^{(5/2) \log(1/\delta)} \leq \delta,\]

as long as $q \leq 2^{-12/5} \cdot $.

Let us call for convenience the above pairs good.
Thus, with probability $1-\delta$ the number of pairs $(j,t)$ for which $j$ is \textbf{not} recognised in repetition $t$ is at most 
\[	\gamma \cdot 4 \log(1/\delta) \cdot \| (x\star y ) - w \|_0  +  \log(1/\delta) \| (x\star y ) - w \|_0 .	\] Hence there exist at most $\beta = 4 \gamma \| (x\star y ) - w \|_0$ indices which are recognized in less than $(3/4) \cdot 5 \log(1/\delta)$ repetitions, otherwise the number of \textbf{not} good pairs $(j,t)$ is at least  
\begin{flalign*}
&1 + \beta\cdot \frac{1}{4} \cdot 5\log(1/\delta) \|(x \star y) - w\|_0 > \\
&\gamma \cdot 4 \log(1/\delta) \cdot \| (x\star y ) - w \|_0  + \log(1/\delta) \| (x\star y ) - w \|_0\end{flalign*}

which does not hold for $\|(x\star y) -w \|_0 > 0$. Moreover, there can be at most $\gamma \|(x\star y ) - w\|_0$ indices that do not belong in $\mathrm{supp}((x \star y ) - w )$, and which were mistakenly inserted into $z$. This gives in total the factor of $5\gamma$. 
\end{proof}

\begin{lemma} \label{lem:bound}
Let $\gamma < 1/10$, and let also $B$ be an integer such that $ B > C \|(x \star y)\|_0$. Then the routine $\textsc{HashAndIterate}(x,y,B,\delta)$ returns an $\|x \star y \|_0$-sparse vector $r$ such that $r = x \star y$, with probability $1-\delta$. Moreover, the running time is 
\begin{align*}
& (B\log^2 N \cdot \log (B \log^2 N) + \\
& (\|x\|_0  + \|y\|_0) \log N \log B )\cdot \log(\log B/\delta).	
\end{align*}
\end{lemma}

\begin{proof}It is an easy induction to show that at each step
$\|(x \star y) - w^{(r)}\|_0 \leq (5 \gamma)^r \|x\star y\|_0$, with probability $ 1-\delta r / \log B$, so the total failure probability is $\delta$.
Conditioned on the previous events happending, we have $x \star y - w^{(\lceil \log B \rceil )}$ is the all-zeros vector since $\|x \star y \|_0 \leq (4\gamma)^{\lceil \log B \rceil } \|(x \star y) - w\|_0 < 1$. This gives that $w^{(\lceil \log B \rceil)} = x \star y$.

The running time for $\textsc{Locate}(x,y,w,B_r,\delta_r)$, since $\|w\|_0 \leq 2B$ at all times is 
(ignoring constant factors for ease of exposition)
\begin{align*}
 ( B_r \log^2 N \log (B_r \log^2 N) +\\
 (\|x\|_0 + \|y\|_0 + B) \log N ) \cdot \log (\log B/\delta_r),
\end{align*}

due to Lemma \ref{lem:prepare} and Lemma \ref{lem:binary_search}. 

So the total running time of $\textsc{HashAndIterate}(x,y,B,\delta)$ becomes, by summing over all $\lceil \log B \rceil$ rounds (ignoring constant factors for ease of exposition)

\begin{align*}
& ( B \log^2 N \cdot \log( B \log^2 N) + \\
&  ( \|x \|_0 + \|y\|_0) \cdot\log N \cdot \log B)\cdot \log(\log B/\delta).
\end{align*}

\end{proof}

The following Lemma is a standard claim which follows by the fact that a degree $n$ polynomial over the prime field $\mathbb{Z}_p$ has at most $n$ roots. We give a sketch of the proof. 
\begin{lemma}\label{lem:fingerprinting}
There exists a procedure $\textsc{EqualityTesting}(x,y,w)$, which runs in time $ O(\|x\|_0 + \|y\|_0 + \|w\|_0)\log N\log(1/\delta) + \tilde{O}( \log^2N \cdot \log (N/\delta) \cdot \log(1/\delta))$, and answers whether $x \star y = w $ with probability $1-\delta$.\end{lemma}

\begin{proof}
Let $c'$ large enough. We pick a random prime $p \in[c'N , 2c' N ]$, by picking a random number in that interval and running the Miler-Rabin primality test with target failure probability $\delta/3$. We form polynomials $f_x,f_y,f_w$ that have $x,y,w$ as their coefficients respectively. We then pick $\Theta(\log (3/\delta))$ random elements in $\mathbb{Z}_p$ and check whether $(f_x(r) \cdot f_y(r) )~\mathrm{mod}~p = f_w(r)~\mathrm{mod}~p$ or not. We return \textsc{Yes} if this is the case for all chosen, and \textsc{No} otherwise.  To evalute each each of the polynomials we need time (number of coefficients) $ \cdot \log N$, in order to perform repetated squaring. The Miller-Rabin test takes time $\tilde{O}(\log^2 N \cdot \log (1/\delta))$, and sampling a prime with probability $1-\delta/3$ needs time $\Theta(\log (N/\delta))$.

\end{proof}

We are now ready to prove our main theorem.

\begin{proof}

Let $c$ be a sufficiently small constant and $C$ a sufficiently large constant. For $r=1,2,\ldots$, one by one we set $B_r \leftarrow C \cdot 2^r$ and $\delta_r = c\cdot r^{-2}$, run $\textsc{HashAndIterate}(x,y,B_r,\delta_r)$ to obtain $z$, and feed it to  $\textsc{EqualityTesting}(x,y,z, cr^{-2})$. We stop when the latter procedure returns \textsc{Yes}. The total failure probability thus is at most

\[	\underbrace{\sum_{r \geq 1 } cr^{-2}}_{\textsc{EqualityTesting}} + \sum_{r \geq 1 } \delta_r = 2\sum_{ r \geq 1 } c r^{-2} \leq \frac{1}{100}.	\]

Conditioned on the aforementioned event happening, the total running time is (ignoring constants) 
\begin{align*}
& \|x \star y\|_0 \log^2 N \cdot \log( \|x \star y\|_0 \log^2 N) \cdot \log \log \|x \star y\|_0 +\\
& (\|x\|_0 + \|y\|_0) \log\|x \star y \|_0 \log N \cdot (\log \log \|x \star y\|_0)^2 + \\
&\tilde{O}(\log^3 N)\cdot \log (\| x \star y\|_0) \log \log (\|x \ast y\|_0). 
\end{align*}

The running time in the first two lines follows by invoking Lemma \ref{lem:bound} and a straightforward summation over all $O( \log (\|x \star y\|_0) )$ rounds, and the third line is the cost of invoking Lemma \ref{lem:fingerprinting} in every round.
 
The final expression now follows by recalling that $N = 2n$.
\end{proof}







\addcontentsline{toc}{section}{References}
\bibliographystyle{alpha}
\bibliography{ref}

\begin{thebibliography}{HIKP12b}

\bibitem[AR15]{arnold2015output}
Andrew Arnold and Daniel~S Roche.
\newblock Output-sensitive algorithms for sumset and sparse polynomial
  multiplication.
\newblock In {\em Proceedings of the 2015 ACM on International Symposium on
  Symbolic and Algebraic Computation}, pages 29--36. ACM, 2015.

\bibitem[CH02]{cole2002verifying}
Richard Cole and Ramesh Hariharan.
\newblock Verifying candidate matches in sparse and wildcard matching.
\newblock In {\em Proceedings of the thiry-fourth annual ACM symposium on
  Theory of computing}, pages 592--601. ACM, 2002.

\bibitem[CL15]{CL15}
Timothy~M. Chan and Moshe Lewenstein.
\newblock Clustered integer {3SUM} via additive combinatorics.
\newblock In {\em Proc. of the 47th Annual {ACM} {S}ymposium on {T}heory of
  {C}omputing ({STOC})}, pages 31--40, 2015.

\bibitem[CRT06a]{crt06a}
Emmanuel~J Cand{\`e}s, Justin Romberg, and Terence Tao.
\newblock Robust uncertainty principles: Exact signal reconstruction from
  highly incomplete frequency information.
\newblock {\em IEEE Transactions on information theory}, 52(2):489--509, 2006.

\bibitem[CRT06b]{crt06}
Emmanuel~J Candes, Justin~K Romberg, and Terence Tao.
\newblock Stable signal recovery from incomplete and inaccurate measurements.
\newblock {\em Communications on pure and applied mathematics},
  59(8):1207--1223, 2006.

\bibitem[CT06]{ct06}
Emmanuel~J Candes and Terence Tao.
\newblock Near-optimal signal recovery from random projections: Universal
  encoding strategies?
\newblock {\em IEEE transactions on information theory}, 52(12):5406--5425,
  2006.

\bibitem[DR93]{dutt1993fast}
Alok Dutt and Vladimir Rokhlin.
\newblock Fast fourier transforms for nonequispaced data.
\newblock {\em SIAM Journal on Scientific computing}, 14(6):1368--1393, 1993.

\bibitem[GLPS10]{glps12}
Anna~C Gilbert, Yi~Li, Ely Porat, and Martin~J Strauss.
\newblock Approximate sparse recovery: optimizing time and measurements.
\newblock {\em SIAM Journal on Computing 2012 (A preliminary version of this
  paper appears in STOC 2010)}, 41(2):436--453, 2010.

\bibitem[GMS05]{gms05}
Anna~C Gilbert, S~Muthukrishnan, and Martin Strauss.
\newblock Improved time bounds for near-optimal sparse {F}ourier
  representations.
\newblock In {\em Optics \& Photonics 2005}, pages 59141A--59141A.
  International Society for Optics and Photonics, 2005.

\bibitem[GR16]{groves2016sparse}
A~Whitman Groves and Daniel~S Roche.
\newblock Sparse polynomials in flint.
\newblock {\em ACM Communications in Computer Algebra}, 50(3):105--108, 2016.

\bibitem[HIKP12a]{hikp12a}
Haitham Hassanieh, Piotr Indyk, Dina Katabi, and Eric Price.
\newblock Nearly optimal sparse fourier transform.
\newblock In {\em Proceedings of the forty-fourth annual ACM symposium on
  Theory of computing}, pages 563--578. ACM, 2012.

\bibitem[HIKP12b]{hikp12b}
Haitham Hassanieh, Piotr Indyk, Dina Katabi, and Eric Price.
\newblock Simple and practical algorithm for sparse {F}ourier transform.
\newblock In {\em Proceedings of the twenty-third annual ACM-SIAM symposium on
  Discrete Algorithms}, pages 1183--1194. SIAM, 2012.

\bibitem[Kap16]{k16}
Michael Kapralov.
\newblock Sparse {F}ourier transform in any constant dimension with
  nearly-optimal sample complexity in sublinear time.
\newblock In {\em Symposium on Theory of Computing Conference, STOC'16,
  Cambridge, MA, USA, June 19-21, 2016}, 2016.

\bibitem[Kap17]{k17}
Michael Kapralov.
\newblock Sample efficient estimation and recovery in sparse fft via isolation
  on average.
\newblock In {\em Foundations of Computer Science (FOCS), 2017 IEEE 58th Annual
  Symposium on}, pages 651--662. Ieee, 2017.

\bibitem[Maz01]{maza2001sparse}
Marc~Moreno Maza.
\newblock Sparse polynomial arithmetic with the bpas library.
\newblock {\em Computer Algebra in Scientific Computing: CASC 2001}, 11077:32,
  2001.

\bibitem[MP07]{monagan2007polynomial}
Michael Monagan and Roman Pearce.
\newblock Polynomial division using dynamic arrays, heaps, and packed exponent
  vectors.
\newblock In {\em International Workshop on Computer Algebra in Scientific
  Computing}, pages 295--315. Springer, 2007.

\bibitem[MP09]{monagan2009parallel}
Michael Monagan and Roman Pearce.
\newblock Parallel sparse polynomial multiplication using heaps.
\newblock In {\em Proceedings of the 2009 international symposium on Symbolic
  and algebraic computation}, pages 263--270. ACM, 2009.

\bibitem[MP14]{monagan2014poly}
Michael Monagan and Roman Pearce.
\newblock Poly: A new polynomial data structure for maple 17.
\newblock In {\em Computer Mathematics}, pages 325--348. Springer, 2014.

\bibitem[MP15]{monagan2015design}
Michael Monagan and Roman Pearce.
\newblock The design of maple's sum-of-products and poly data structures for
  representing mathematical objects.
\newblock {\em ACM Communications in Computer Algebra}, 48(3/4):166--186, 2015.

\bibitem[Roc08]{roche2008adaptive}
Daniel~S Roche.
\newblock Adaptive polynomial multiplication.
\newblock {\em Proc. Milestones in Computer Algebra (MICA’08)}, pages 65--72,
  2008.

\bibitem[Roc11]{roche2011chunky}
Daniel~S Roche.
\newblock Chunky and equal-spaced polynomial multiplication.
\newblock {\em Journal of Symbolic Computation}, 46(7):791--806, 2011.

\bibitem[Roc18]{roche2018can}
Daniel~S Roche.
\newblock What can (and can't) we do with sparse polynomials?
\newblock {\em arXiv preprint arXiv:1807.08289}, 2018.

\bibitem[VDHL12]{van2012complexity}
Joris Van Der~Hoeven and Gr{\'e}goire Lecerf.
\newblock On the complexity of multivariate blockwise polynomial
  multiplication.
\newblock In {\em Proceedings of the 37th International Symposium on Symbolic
  and Algebraic Computation}, pages 211--218. ACM, 2012.

\bibitem[VDHL13]{van2013bit}
Joris Van Der~Hoeven and Gr{\'e}goire Lecerf.
\newblock On the bit-complexity of sparse polynomial and series multiplication.
\newblock {\em Journal of Symbolic Computation}, 50:227--254, 2013.

\bibitem[Yan98]{yan1998geobucket}
Thomas Yan.
\newblock The geobucket data structure for polynomials.
\newblock {\em Journal of Symbolic Computation}, 25(3):285--293, 1998.

\end{thebibliography}





\end{document}